\documentclass[aps,prl,twocolumn,superscriptaddress,longbibliography]{revtex4-2}

\usepackage{amsmath,amssymb,amsthm}
\usepackage[colorlinks=true,linkcolor=blue,citecolor=blue,urlcolor=blue]{hyperref}

\newtheorem{lemma}{Lemma}
\newtheorem{theorem}{Theorem}

\newcommand{\Tr}{\operatorname{Tr}}
\newcommand{\nh}{\hat n}
\newcommand{\nb}{\bar n}

\begin{document}

\title{Pure states are optimal for total energy-constrained quantum channel discrimination}

\author{Christos N. Gagatsos}
\affiliation{School of Electrical, Computing, and Software Engineering, The University of Arizona, Tucson, Arizona 85721, USA}
\affiliation{Wyant College of Optical Sciences, The University of Arizona, Tucson, Arizona 85721, USA}
\affiliation{Program in Applied Mathematics, The University of Arizona, Tucson, Arizona 85721, USA}

\date{\today}

\begin{abstract}
We consider the discrimination of two quantum channels acting on bosonic modes, with a probe state of fixed mean photon number and no ancilla.
We prove that pure probes are optimal: for every mixed probe there is a pure probe with exactly the same mean photon number that performs at least as well.
The result holds for arbitrary channels and for every figure of merit built from the Chernoff quantities $\Tr[\rho_1^s\rho_2^{1-s}]$, including the quantum Chernoff exponent, as well as for the single-shot error probability.
The proof relies on a decomposition of any state of mean photon number $\nb$ into pure states of mean photon number exactly $\nb$, which extends a matrix rank-one decomposition of Huang and Zhang to density operators of infinite rank.
\end{abstract}

\maketitle

\textit{Introduction.}---
Deciding which of two physical processes has acted on a probe is a basic task behind a large part of quantum information science and of its applications.
Discriminating unitaries~\cite{Acin2001,DAriano2001}, detecting a target with quantum illumination~\cite{Lloyd2008,Tan2008}, reading a classical memory~\cite{Pirandola2011,Nair2011}, sensing loss and imaging~\cite{NairYen2011,Invernizzi2011}, deciding between one and two point sources below the diffraction limit~\cite{Lu2018}, and detecting a change point in a sequence of quantum processes~\cite{Sentis2016} are all instances of channel discrimination.
The same task underlies remote sensing, optical data readout, and the certification of quantum devices, and its ultimate limits have been studied intensively~\cite{Pirandola2019,Zhuang2020,Becker2021,Huang2026}.

A natural question is which probe state of fixed mean photon number $\nb$ discriminates the two channels best.
When the error probability decays with the number $M$ of channel uses as $e^{-M\xi_Q}$, the relevant figure of merit is the quantum Chernoff exponent (QCE) $\xi_Q$~\cite{Audenaert2007,Nussbaum2009}.

It is widely expected, and often assumed, that the optimal probe can be taken pure.
If an ancilla is available and only the photons sent through the channel are counted, this is easy to show: a mixed probe can be purified with the ancilla without changing the signal state, and discarding the purifying system can only degrade performance~\cite{NairYen2011,Huang2026}.
Without an ancilla the expectation rests on convexity: a mixed probe $\rho=\sum_jp_j|\phi_j\rangle\langle\phi_j|$ does not perform better than its best pure component.
However, the energy constraint imposes a complication:
the pure components $|\phi_j\rangle$ have in general mean photon numbers different from $\nb$, some larger and some smaller, so the best component may simply not be an admissible probe.
The same gap appears in quantum metrology, where convexity of the quantum Fisher information is used to restrict to pure probes at fixed energy~\cite{GongGuha2026,Maliakal2026}.

In this work we close this gap.
We show that every state of mean photon number $\nb$ can be written as a mixture of pure states, each of which has mean photon number \emph{exactly} $\nb$ (Lemma~\ref{lem:decomp}).
For finite rank this is a complex-Hermitian version of the matrix rank-one decomposition of Sturm and Zhang~\cite{SturmZhang2003} and Huang and Zhang~\cite{HuangZhang2007}; we extend it to density operators of infinite rank, i.e., continuous variables states, by an iteration that we control in the limit.
With this decomposition we then proceed to the main reult given in Theorem~\ref{thm:main}.

\textit{Setting.}---
The probe is a state $\rho$ of $m$ bosonic modes with annihilation operators $\hat a_1,\dots,\hat a_m$.
We write $\nh=\sum_{l=1}^m\hat a_l^\dagger\hat a_l$ for the total photon number and impose
\begin{equation}
\Tr[\nh\,\rho]=\nb ,\qquad 0\le\nb<\infty .
\label{eq:constraint}
\end{equation}
The probe is sent through one of two quantum channels, $\mathcal N_1$ or $\mathcal N_2$ (completely positive and trace preserving maps), and the receiver has to decide which one acted.
For two states $\sigma_1,\sigma_2$ let
\begin{equation}
Q_s=\Tr[\sigma_1^s\sigma_2^{1-s}],\qquad Q=\inf_{0\le s\le 1}Q_s,\qquad \xi_Q=-\ln Q ,
\label{eq:qce}
\end{equation}
where $\sigma^0$ denotes the projector onto the support of $\sigma$~\cite{Audenaert2007,Nussbaum2009}.
For a probe $\rho$ we write $Q_s(\rho)$, $Q(\rho)$ and $\xi_Q(\rho)$ for these quantities evaluated on $\sigma_i=\mathcal N_i(\rho)$.
With $M$ uses of the channel and the product probe $\rho^{\otimes M}$, $Q_s$ is multiplicative, so $\xi_Q(\rho)$ is the exponent per channel use.

\textit{Equal-photon-number decomposition.}--- We state and prove the following lemma:

\begin{lemma}\label{lem:decomp}
Let $\rho$ be a state on $m$ modes with $\Tr[\nh\rho]=\nb<\infty$.
Then there are countably many pure states $|\psi_j\rangle$ and weights $w_j>0$, $\sum_jw_j=1$, such that
\begin{equation}
\rho=\sum_j w_j\,|\psi_j\rangle\langle\psi_j|,\qquad \langle\psi_j|\nh|\psi_j\rangle=\nb\quad\text{for all } j .
\label{eq:decomp}
\end{equation}
\end{lemma}

\begin{proof}
Start from the spectral decomposition $\rho=\sum_j|v_j\rangle\langle v_j|$, with unnormalized eigenvectors, $\langle v_j|v_j\rangle=\lambda_j$.
Define the following real number (which can be viewed as photon-number excess),
\begin{equation}
a_j=\langle v_j|(\nh-\nb)|v_j\rangle .
\label{eq:excess}
\end{equation}
Constraint~\eqref{eq:constraint} reads as $\sum_ja_j=0$, and the sum converges absolutely since $|a_j|\le\langle v_j|\nh|v_j\rangle+\nb\lambda_j$.
Vectors with $a_j=0$ already have mean photon number $\nb$ once normalized, and we set them aside.
The others go into two lists, $\mathcal P$ (with $a_j>0$) and $\mathcal N$ (with $a_j<0$), each kept in the original order.

\emph{Elementary step.}
Take $|u\rangle$ from $\mathcal P$ and $|v\rangle$ from $\mathcal N$, with (defined in Eq. \eqref{eq:excess}) $a>0$ and $b<0$.
Let $c=\langle u|(\nh-\nb)|v\rangle$, which is finite by the Cauchy--Schwarz inequality, and choose a phase $\varphi$ with $\mathrm{Re}(e^{i\varphi}c)=0$.
Define
\begin{align}
|w\rangle&=\cos\theta\,|u\rangle+e^{i\varphi}\sin\theta\,|v\rangle,\nonumber\\
|w'\rangle&=-\sin\theta\,|u\rangle+e^{i\varphi}\cos\theta\,|v\rangle,\qquad \tan^2\theta=\frac{a}{|b|} .
\label{eq:rotation}
\end{align}
This is a unitary rotation of the pair, so $|w\rangle\langle w|+|w'\rangle\langle w'|=|u\rangle\langle u|+|v\rangle\langle v|$.
The cross terms vanish by the choice of $\varphi$, and
\begin{align}
\langle w|(\nh-\nb)|w\rangle&=a\cos^2\theta+b\sin^2\theta=0,\nonumber\\
\langle w'|(\nh-\nb)|w'\rangle&=a+b,
\label{eq:step}
\end{align}
where we used the choice of $\theta$ given in Eq. \eqref{eq:rotation}. Therefore, $|w\rangle$ has mean photon number exactly $\nb$, and $|w'\rangle$ carries the remaining excess $a+b$.

\emph{Iteration.}
Apply the previous step to the first elements of $\mathcal P$ and $\mathcal N$.
Set $|w\rangle$ aside.
The vector $|w'\rangle$ replaces $|u\rangle$ at the head of $\mathcal P$ if $a+b>0$, replaces $|v\rangle$ at the head of $\mathcal N$ if $a+b<0$, and is set aside if $a+b=0$; the other member of the pair has been used already. Then
repeat:
each step uses up one of the original vectors.
Moreover, every original vector is used after finitely many steps.
For instance, the head of $\mathcal P$ has an excess equal to a partial sum of the positive $a_j$ plus a partial sum of the negative $a_j$; since all $a_j$ add up to zero, this excess becomes non-positive after finitely many steps as long as further vectors remain in $\mathcal P$, and the list moves on.

\emph{Limit.}
After $k$ steps we have
\begin{equation}
\rho=F_k+|c_k\rangle\langle c_k|+T_k ,
\label{eq:partial}
\end{equation}
where $F_k$ is the sum of the vectors set aside, all with mean photon number $\nb$, $|c_k\rangle$ is the current carried vector (possibly zero), and $T_k$ is the sum of the original vectors not yet used.
As $k\to\infty$, $T_k\to0$ because it is the tail of a convergent sum.
The operators $F_k$ grow with $k$ and are bounded by $\rho$, so they converge to $F=\sum_j|z_j\rangle\langle z_j|$.
Then $|c_k\rangle\langle c_k|$ converges as well, to an operator $|c\rangle\langle c|$ of rank at most one, and $\rho=\sum_j|z_j\rangle\langle z_j|+|c\rangle\langle c|$.
It remains to check the carried vector, since its excess tends to zero but energy could in principle escape in the limit.
Because $\nh\ge0$, every term contributes non-negatively to the mean photon number, and
\begin{align}
\nb=\Tr[\nh\rho]&=\sum_j\langle z_j|\nh|z_j\rangle+\langle c|\nh|c\rangle\nonumber\\
&=\nb\sum_j\langle z_j|z_j\rangle+\langle c|\nh|c\rangle .
\label{eq:additivity}
\end{align}
Since $\sum_j\langle z_j|z_j\rangle+\langle c|c\rangle=\Tr\rho=1$, this gives $\langle c|\nh|c\rangle=\nb\,\langle c|c\rangle$: the carried vector, if non-zero, also has mean photon number $\nb$.
Normalizing all vectors, with $w_j=\langle z_j|z_j\rangle$, gives Eq.~\eqref{eq:decomp}.
\end{proof}

For a state of finite rank $r$ the iteration stops after at most $r-1$ steps, and Lemma~\ref{lem:decomp} reduces to the rank-one decomposition of Refs.~\cite{SturmZhang2003,HuangZhang2007}.
The decomposition is not unique.
Note also that nothing in the proof is specific to $\nh$: any Hamiltonian bounded from below, for example $\sum_l\omega_l\hat a_l^\dagger\hat a_l$ with unequal frequencies, can be used instead.

\textit{Main result.}---The following the theorem is the main result:
\begin{theorem}\label{thm:main}
Let $\mathcal N_1,\mathcal N_2$ be any two quantum channels acting on $m$ bosonic modes, and let $\nb\ge0$.
Then
\begin{equation}
\sup_{\rho:\ \Tr[\nh\rho]=\nb}\ \xi_Q(\rho)\ =\sup_{|\psi\rangle:\ \langle\psi|\nh|\psi\rangle=\nb}\ \xi_Q(\psi) .
\label{eq:main}
\end{equation}
More precisely, if the infimum over $s$ in $Q(\rho)$ is attained, one of the pure states $|\psi_j\rangle$ of Eq.~\eqref{eq:decomp} satisfies $\xi_Q(\psi_j)\ge\xi_Q(\rho)$.
\end{theorem}

\begin{proof}
Write $\rho=\sum_jw_j|\psi_j\rangle\langle\psi_j|$ as in Lemma~\ref{lem:decomp}.
The channels are linear, so $\mathcal N_i(\rho)=\sum_jw_j\mathcal N_i(\psi_j)$.
For $0<s<1$ the map $(A,B)\mapsto\Tr[A^sB^{1-s}]$ is jointly concave (Lieb's theorem~\cite{Lieb1973}), hence $Q_s(\rho)$ is a concave function of $\rho$.
Jensen's inequality then gives
\begin{equation}
Q_s(\rho)\ \ge\ \sum_jw_j\,Q_s(\psi_j) .
\label{eq:jensen}
\end{equation}
The sum in Eq.~\eqref{eq:jensen} may contain infinitely many terms; this does not cause a difficulty because $Q_s$ depends continuously on the probe. Truncating the decomposition after finitely many terms and renormalizing gives a state arbitrarily close to $\rho$ in trace norm. Jensen's inequality holds for the truncated mixture, and since $Q_s$ changes by an arbitrarily small amount when the probe does~\cite{Ando1988,Bhatia1997}, the inequality survives the limit.
At the end points the same inequality holds directly.
For $s=0$, $Q_0(\rho)=\Tr[\Pi\,\mathcal N_2(\rho)]$, with $\Pi$ the projector onto the support of $\mathcal N_1(\rho)$.
This support contains the support of every $\mathcal N_1(\psi_j)$, so $\Pi$ is larger than each of the corresponding projectors and Eq.~\eqref{eq:jensen} follows; the case $s=1$ is analogous.

A weighted average cannot be smaller than all of its terms, so for every $s$ there is a $j$ with $Q_s(\psi_j)\le Q_s(\rho)$.
Let $s^*$ attain the infimum in $Q(\rho)$, and pick $j$ for $s=s^*$.
Then
\begin{equation}
Q(\psi_j)\ \le\ Q_{s^*}(\psi_j)\ \le\ Q_{s^*}(\rho)\ =\ Q(\rho) ,
\label{eq:chain}
\end{equation}
that is, $\xi_Q(\psi_j)\ge\xi_Q(\rho)$, and $|\psi_j\rangle$ has mean photon number $\nb$.
If the infimum is not attained, the same argument applied to an $s$ with $Q_s(\rho)\le Q(\rho)+\epsilon$ gives a pure state with $Q(\psi_j)\le Q(\rho)+\epsilon$ for every $\epsilon>0$, which proves Eq.~\eqref{eq:main}.
\end{proof}

The argument uses nothing about the channels except linearity, so it applies to Gaussian and non-Gaussian channels alike, to any number of output modes, and to channels with different output spaces.
It also applies to other figures of merit.
The single-shot error probability with priors $p$ and $1-p$ is $P_e=\frac12\big(1-\|p\,\mathcal N_1(\rho)-(1-p)\,\mathcal N_2(\rho)\|_1\big)$~\cite{Helstrom1976}, and by the triangle inequality $\|p\,\mathcal N_1(\rho)-(1-p)\,\mathcal N_2(\rho)\|_1\le\sum_jw_j\|p\,\mathcal N_1(\psi_j)-(1-p)\,\mathcal N_2(\psi_j)\|_1$, so again one of the $|\psi_j\rangle$ does at least as well.
The finite-$M$ Chernoff bound with unequal priors, $\min_sp^s(1-p)^{1-s}Q_s(\rho)^M$, and the error exponents obtained when one prior decays exponentially with $M$ are built from a single $Q_s$ at the optimal $s$, and the same chain~\eqref{eq:chain} applies.

\textit{Conclusion.}---
We have shown that, for the discrimination of any two quantum channels with a probe of fixed mean photon number and without an ancilla, pure probes are optimal.
The proof combines the concavity of the Chernoff quantities with a decomposition of any state into pure states of exactly the same mean photon number, valid for density operators of infinite rank.
The result justifies the restriction to pure probes that is common in quantum illumination, quantum reading, loss and phase sensing, and related tasks, and it reduces the search for optimal probes to pure states without any loss in the energy budget.
Which pure states are optimal depends on the channels, and identifying them, in particular within physically accessible classes, remains the task for each application. 

Three remarks must be made at this point: 
\begin{enumerate}
    \item If an ancilla is available and the constraint counts only the photons sent through the channel, pure probes are optimal by purification~\cite{NairYen2011,Huang2026}.
If instead the constraint counts the photons in the probe and in the ancilla together, purification costs photons and is no longer free.
Theorem~\ref{thm:main} covers this case: one applies it to the joint system, with the channels $\mathcal N_i\otimes\mathrm{id}$ and $\nh$ the total photon number.

\item Theorem~\ref{thm:main} concerns the optimization over all states.
It does not imply that pure states are optimal within a restricted class, because the pure states $|\psi_j\rangle$ of Lemma~\ref{lem:decomp} need not belong to that class.
For number-diagonal probes, for instance, no pure state (in that case, a Fock state) of non-integer $\nb$ exists at all.
For Gaussian probes the decomposition of a mixed Gaussian state produces non-Gaussian pure states, and the optimality of pure Gaussian probes among Gaussian ones would require a separate analysis. The optimal pure probe of Theorem~\ref{thm:main} may well be non-Gaussian~\cite{Nair2011}.

\item It would be interesting to rigorously see if Lemma~\ref{lem:decomp} applies to other figures of merit that are convex in the probe, provided Jensen's inequality can be used for the resulting mixture.
For probes of finite rank only finitely many terms appear and no continuity is needed.
This includes the quantum Fisher information, where the same issue with the energy constraint arises~\cite{GongGuha2026,Maliakal2026}.
\end{enumerate}

\bibliography{refs}

\end{document}